\def\year{2018}\relax
\documentclass[letterpaper]{article} 
\usepackage{aaai18}  
\usepackage{times}  
\usepackage{helvet}  
\usepackage{courier}  
\usepackage{url}  
\usepackage{graphicx}  
\usepackage{subcaption}
\usepackage{dsfont}
\usepackage{multicol}
\makeatletter
\newif\if@restonecol
\makeatother

\usepackage[lined, boxed, ruled, commentsnumbered]{algorithm2e}
\usepackage{amsmath,amssymb,xfrac,amsthm}
\newtheorem{lemma}{Lemma}
\newtheorem{theorem}{Theorem}
\newtheorem*{theorem*}{Theorem}

\newtheorem{corollary}{Corollary}
\makeatletter

\newcommand{\Rmnum}[1]{\expandafter\@slowromancap\romannumeral #1@}
\makeatother
\usepackage{caption}

\usepackage{algpseudocode}
\usepackage{enumitem}
\usepackage{wasysym}

\DeclareMathOperator*{\argmax}{arg\,max}

\newcommand{\R}{\mathbb{R}}

\newcommand{\pay}{\tau}
\newcommand{\E}{\mathds E}

\newcommand{\feasible}{\phi}

\newcommand{\budget}{B}
\newcommand{\groundSet}{V}
\newcommand{\selectSet}{S}

\newcommand{\selectSetAlg}{\ensuremath{\widetilde{S^{*}}}}
\newcommand{\funcf}{f}

\newcommand{\neighbor}{\ensuremath{\mathcal{N}}}
\newcommand{\graph}{G}
\newcommand{\edges}{E}

\newcommand{\opt}{\textsc{Opt}\xspace}

\newcommand{\ppcgreedyalg}{\textsc{PPCGreedy}\xspace}
\newcommand{\ppcgreedyalgiter}{\textsc{PPCGreedyIter}\xspace}

\newcommand{\supergroundsetalg}{\textsc{CoupleSuperSet}\xspace}
\newcommand{\supergroundset}{\ensuremath{\mathcal{Z}}}

\newcommand{\citet}[1]{\citeauthor{#1} (\citeyear{#1})}

\begin{document}
\title{Information Gathering with Peers: \\Submodular Optimization with Peer-Prediction Constraints}
\author{
Goran Radanovic\\
Harvard University\\
Cambridge, USA\\
gradanovic@g.harvard.edu\\
\And
Adish Singla\\
MPI-SWS\\
Saarbr{\"u}cken, Germany\\
adishs@mpi-sws.org\\
\And
Andreas Krause\\
ETH Zurich\\
Zurich, Switzerland\\
krausea@ethz.ch\\
\And
Boi Faltings\\
EPFL\\
Lausanne, Switzerland\\
boi.faltings@epfl.ch\\
}

\maketitle
\begin{abstract}
We study a problem of optimal information gathering from multiple data providers that need to be incentivized to provide accurate information. This 
problem arises in many real world applications that rely on crowdsourced data sets, but where the process of obtaining data is costly. A notable 
example of such a scenario is crowd sensing. To this end, we formulate the problem of optimal information gathering as maximization of a submodular
function under a budget constraint, where the budget represents the total expected payment to data providers. Contrary to the existing approaches, we base our 
payments on incentives for accuracy and truthfulness, in particular, {\em peer-prediction} methods that score each of the selected data providers against its best peer, while 
ensuring that the minimum expected payment is above a given threshold. We first show that the problem at hand is hard to approximate within a constant 
factor that is not dependent on the properties of the payment function. However, for given topological and analytical properties of the 
instance, we construct two greedy algorithms, respectively called PPCGreedy and PPCGreedyIter, and establish theoretical bounds on their 
performance w.r.t. the optimal solution. Finally, we evaluate our methods using a realistic crowd sensing testbed.  
\end{abstract}

\section{Introduction}\label{sec.introduction}

The recent success of various machine learning techniques can partly be attributed to
the existence of large sets of labeled data that can readily be used for training purposes. In the past decade, the 
predominant form of obtaining useful data is through crowdsourcing approaches, where human subjects either directly 
label data or have private devices that provide measurements about spatially distributed phenomena. 

\begin{figure}[!h]
\centering
\includegraphics[width=0.8\columnwidth]{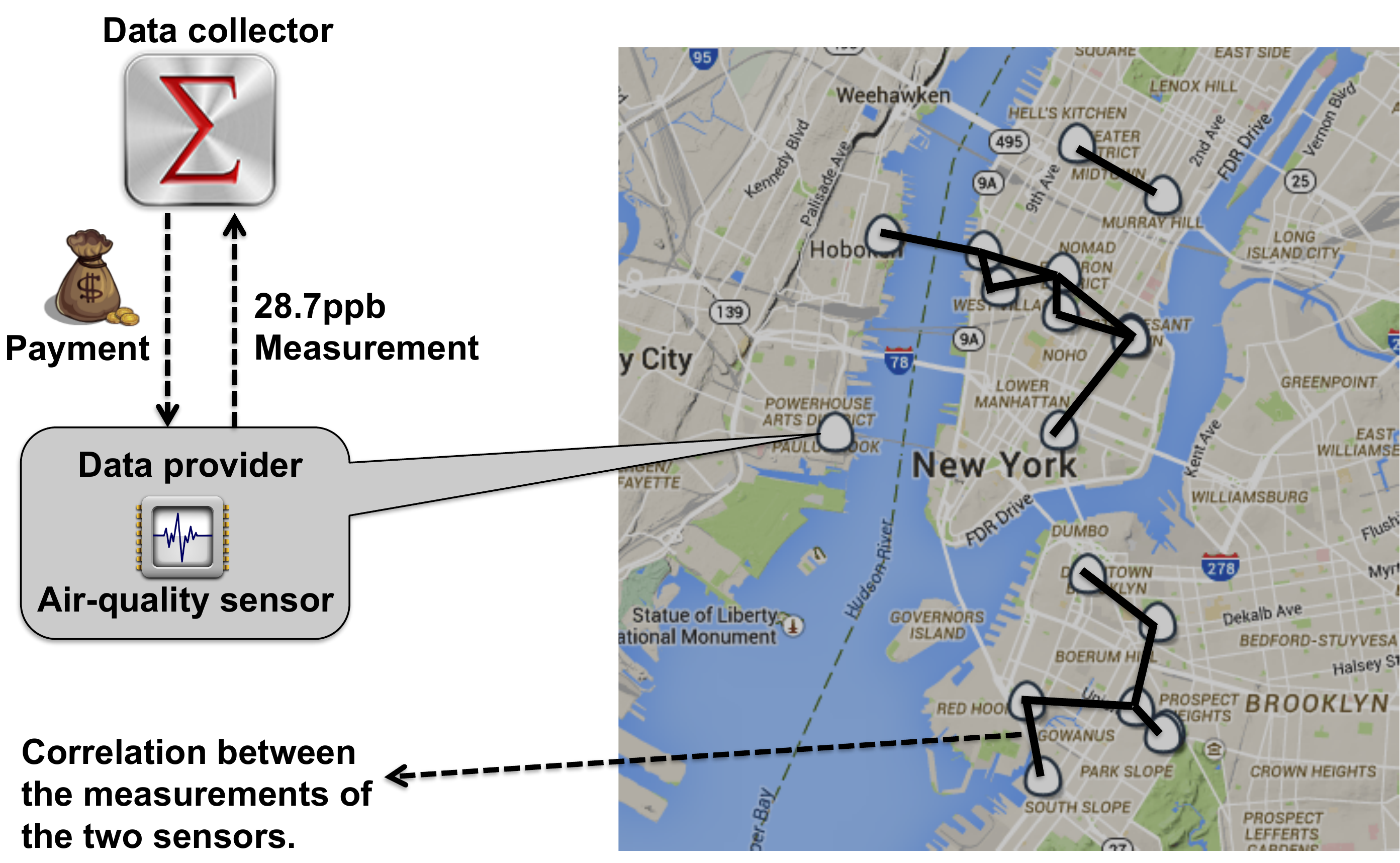}
\caption[Fig:sensing:scenario]{An example of crowdsourcing with incentives: Crowd-sensors ({\em air-quality eggs}\footnotemark[1])
report air-quality measurements and a data collector rewards them with monetary payments. 
An edge (black line) indicates that there is a sufficient correlation between the measurements of two sensors to verify accuracy. }
\label{fig:sensing}
\end{figure}

One of the most important aspects of data is its accuracy, which 
can only be established if data providers (e.g. crowd-participants) report accurate information. To incentivize accurate reporting, 
a data collector can provide incentives that compensate data providers for their effort. 
In its simplest form, this type of data elicitation process can be modeled as a three step protocol:
\begin{itemize}
\item Data providers acquire accurate data experiencing a cost of effort; 
\item Data providers report the acquired data to a data collector;
\item The data collector pays to the data providers a value that compensates for the cost of effort. 
\end{itemize}
An example scenario is shown in Figure \ref{fig:sensing}.\footnotetext[1]{https://airqualityegg.wickeddevice.com}

The problem, however, arises if the data providers are susceptible to moral hazard, that is, deviation to reporting {\em heuristically} without obtaining 
the data in the first place. In fact, such a behavior is expected for a rational participant who aims to maximize their utility, since heuristic reporting typically
carries no cost of effort. To avoid this problem, the data collector can design payment functions $\tau$ that are dependent on the accuracy of the reported 
information, for example, by installing random spot-checks that validate some of the reports \cite{DBLP:journals/corr/GaoWL16}. 
While this approach has often been used in standard micro-task crowdsourcing, it is often too costly to apply it in a more complex elicitation setting. Consider, for 
example, a crowd-sensing scenario shown in Figure \ref{fig:sensing}, where sensors measure spatially distributed phenomenon that is highly localized. 
To apply a spot-checking procedure, the data collector would need to have mobile sensors that would change their locations at each time-step. Furthermore, due to 
the localized nature of the measured phenomenon, the density of the spot-check sensor network has to be relatively large. 

Instead of evaluating data providers against trusted reports,  
\citet{DG:13}, \citet{JF:11}, \citet{RF:16b}, \citet{Shnayder:2016:EC}, \citet{DBLP:conf/aaai/WitkowskiAUK17}, and \citet{A:17}, 
propose {\em peer-prediction} mechanisms for incentivizing distributed 
information sources. Peer-prediction mechanisms reward data providers by measuring consistency among 
their reports---thus, if a data provider believes that others are honest, she is also incentivized to report truthfully.\footnote{Peer-predictions are in general susceptible to collusion, but in many cases one can establish relatively strong incentive properties for a wide variety of reporting 
strategies \cite{KS:16b}. We do not focus on collusion resistance, so we use standard peer-predictions in our setting.
}  
The most important condition to hold when applying 
a peer-prediction mechanism is that a data provider and her peer have correlated private information. Furthermore, this correlation, when expressed through expected 
payments, should be greater than the cost of effort.  The latter property can always be achieved by scaling, provided that a considered peer-prediction method provides strict incentives for truthfulness. 
The scaling approach, however, neglects potential budget concerns that are important when collecting large data sets.  

\subsection{Overview of Our Approach}\label{sec.introduction.overview}

We, therefore, focus on the limited budget concern in a distributed data collection process that uses peer-prediction incentives. There are two important aspects to this problem:
\begin{itemize}
\item which data providers to select given that we only have a limited budget to spend on incentives---thus, only those data providers who received incentives can be considered to be reliable;
\item how to ensure that all of the selected data providers have a proper peer---this constrains the selection problem to always include a proper peer of each data provider that is to be selected.
\end{itemize}  
To quantify the usefulness of each data provider, we adopt a {\em submodular} utility function, which can, for example, measure the information gain of the data collector for obtaining the reports of the selected data-providers. We will insist that each data provider can be scored against a peer report with resulting expected payment being greater than a given threshold. Furthermore, the total expected payment should be bounded by a budget, while a data provider should be always scored against the best peer among the selected data providers. Our main contributions are:
\begin{itemize}
\item A formal model of information gathering with budget and peer-prediction constraints that is based on submodular maximization.   
\item Showing that the studied optimization problem is hard to approximate within a constant factor independent of the properties of the applied payment function.
\item Novel algorithms for maximizing submodular functions with peer-prediction constraints that have provable guarantees for given topological and analytical properties of payments.
\item Experimental evaluation of the proposed algorithms on a crowd-sensing test-bed. 
\end{itemize}
Notice that we do not focus on a particular peer-prediction mechanism, but rather we allow a wide range of possible mechanisms (those that are robust in terms of the number of peers and produce bounded expected payments); thus, we complement the prior work on peer-predictions by examining orthogonal aspects of elicitation without direct verification. We provide the proofs to our formal claims in the extended version of the paper \cite{RSKF:17}.

\section{Problem Statement}\label{sec.model}

We now formalize the problem addressed in this paper. We model data providers as nodes in a graph, whereas the underlying 
peer-prediction dependencies are modeled via edges whose weights are defined by the expected payments. 
The overall goal is to select a set of nodes that maximize a submodular utility function, while satisfying the constraint 
that the cost of the data collector (i.e., the total expected payment to nodes) 
is within a predefined budget. The following subsections provide more precise modeling details.  

\subsection{Set of Nodes and the Utility Function}
We consider a set of nodes (e.g., a population of people or sensors deployed in a city) denoted by set $\groundSet=\{v_1, v_2,\dots,v_{|\groundSet|}\}$, of size $|\groundSet|$. Hereafter, we denote a generic node by $v$. We associate a function over the set of nodes $\funcf:2^\groundSet \rightarrow \R_{\ge 0}$ that quantifies their utility (e.g., informativeness).  That is, given a set of selected nodes $\selectSet$, the utility achieved from this set is equal to $\funcf(\selectSet)$. Furthermore, for given  set $\selectSet$, and a node $v \in \groundSet \setminus \selectSet$, we define the marginal utility of adding $v$ to $S$ as follows:
\begin{align}
\funcf(v|\selectSet) = \funcf(\selectSet \cup \{v\}) -  \funcf(\selectSet) \label{eq.marginalgain}.
\end{align}

Here, the function $\funcf$ is assumed to be \emph{submodular} and \emph{monotone}. Submodularity is an intuitive notion of diminishing returns, stating that, for any sets $\selectSet \subseteq \selectSet' \subseteq \groundSet$, and any given node $v \notin \selectSet'$, it holds that $f(v|\selectSet) \geq f(v|\selectSet')$. 
Monotonicity requires that the function $\funcf$ increases as we add more elements to set $S$. That is, for any sets $\selectSet \subseteq \selectSet' \subseteq \groundSet$, it holds that $f(\selectSet) \le f(\selectSet')$. These conditions are relatively general, and are satisfied by many realistic, as well as complex utility functions for information gathering \cite{krause2011submodularity,2012-survey_krause_submodular,singla13incentives,singla14stochasticPrivacy,tschiatschek17ordered}. W.l.o.g., we assume that function $f$ is normalized, i.e., $\funcf(\emptyset) = 0$.
 
\subsection{Peer-prediction Constraints (PPC)}
The nodes in general exhibit dependencies with other nodes. We consider a particular form of constraints that is associated with information elicitation via {\em peer-prediction mechanisms} \cite{MRZ:05}. A canonical peer-prediction $\pay: \groundSet \times \groundSet \rightarrow \R$ scores the information reported by node $v$ using the information of a peer node $v_p$. Mechanism $\pay(v,v_p)$ is said to be {\em proper} if node $v$'s best response to accurate reporting of node $v_p$ is to report accurately, where the quality of a response is measured in terms of node $v$'s expected payoff over possible (accurate) reports of node $v_p$.\footnote{Therefore, {\em properness} is here defined as {\em Bayes-Nash incentive compatibility} in game-theoretic sense.} We denote node $v$'s expected payoff for accurate reporting by $\E(\pay(v,v_p))$. To establish the properness of $\pay$, one needs to ensure that peer $v_p$ provides statistically correlated information to that of node $v$,
so that the expected payoff $\E(\pay(v,v_p))$ is strictly greater than the cost of accurate reporting (which models, for example, participants' effort exertion).  
Therefore, node $v$ has only a limited number of peers defined as nodes that lead to the expected payoff $\E(\pay(v,v_p)) \ge \pay_{min}$, where $\tau_{min} > 0$ is a problem specific threshold dependent on the cost of accurate reporting. We will further require that the same holds for node $v$'s peers, i.e., $\E(\pay(v_p,v)) \ge \pay_{min}$, and assume that mechanism $\pay$ provides bounded payments, so that $\E(\pay(v,v_p)) \le \pay_{max}$.
Notice that as $\pay_{min}$ increases, a node is expected to have a smaller number of peers, which makes the problem of selecting an optimal set of nodes more constrained. In Section 'Experimental Evaluation', we confirm this observation by showing the drop in the obtained utility.   

{\bf Example: Output Agreement (OA)}. Arguably the simplest peer prediction method is the output agreement of \citet{vonAhn:2004:LIC:985692.985733}, which gives a strictly positive payment only for matching reports. In our experiments, reported information can in general take real values. In that case, as explained by \citet{DBLP:conf/hcomp/WaggonerC14}, the OA mechanism can be defined as:
\begin{align}\label{eq_OA}
\pay(v, v_p) = 1-d(v,v_p)^2,
\end{align}
where $d$ is Euclidian distance between reported values of $v$ and $v_p$. Note that more complex designs are also allowed by our framework, such as the one proposed in \cite{faltings2013}. For more information on the properties of different minimal peer-predictions and their relationships, we refer the reader to \cite{FW:16}.  

With this in mind, we can model the dependencies among nodes using an undirected graph $\graph=(\groundSet,\edges)$,   
where edges are defined as $\edges = \{\{v, u \}: v, u\in \groundSet, v \ne u, \E(\pay(v,u)) \geq \pay_{min}, \E(\pay(u,v)) \geq \pay_{min} \}$.
We require that each node in a selected set $S$ has a neighboring node $\graph$ that is also in $S$, which implies that we can properly evaluate the reported information of each selected node.
We denote the set of neighboring nodes to node $v$ as $\neighbor_v$, i.e., $\neighbor_v = \{u \in  \groundSet : \{v, u\} \in \edges\}$.
W.l.o.g. we can assume that every node in $\graph$ has at least one peer node $v_p$, i.e., $|\neighbor_v| \ge 1$. Namely, nodes that do not have a peer cannot be incentivized to report accurately in our setting, so they bring $0$ utility in terms of $f$. Finally, let us denote by $\omega$ the maximum number of peers that a node in graph $\graph$ has, i.e.,
$\omega(\graph) = \max_{v \in V} |\neighbor_v|$.
 
\subsection{Cost of Incentivizing Accuracy}
Given a selected set of nodes $S$, an information elicitation procedure needs to spend a certain amount of budget, hereafter denoted by $B$, to incentivizing accurate reporting. To quantify the cost of accurate elicitation, one needs to specify a peer selection procedure when a node $v$ has multiple peers. We take the approach of selecting the best peer, that is the peer that has the most correlated information to that of the considered node according to the expected payoffs---this leads to the strongest incentives in terms of the separation between the expected payoffs for accurate and inaccurate reporting. With this choice of peer selection procedure, we can define the cost of selecting nodes $S$ as function 
$c:2^V_\feasible \rightarrow \R$:
\begin{align}\label{eq_cost_function}
c(S) = \sum_{v \in S} \max_{v_p \in S \cap \neighbor_v} \E(\pay(v, v_p)).
\end{align}
Here, $2^V_\feasible$ contains only sets $S$ such that each node $v \in S$ has a peer node $v_p \in S$,
which makes $c$ well defined. 

\subsection{Optimization Problem}
Our goal is to select a {\em feasible} set $\selectSet \in 2^\groundSet_\feasible$ that maximizes the utility $\funcf(\selectSet)$ given budget $\budget$, i.e., $c(\selectSet) \leq \budget$. More precisely, the budget denotes the total expected payment that a data collector is willing to provide for incentivizing accurate reporting.  
We therefore pose the following optimization problem:
\begin{align}
\selectSet^* = \argmax_{\selectSet \in 2^\groundSet_\feasible \text{ s.t. } c(\selectSet) \leq \budget} \funcf(\selectSet) \label{eq.opt}.
\end{align}
Ignoring the computational constraints, we denote the optimal solution to this problem as $\opt$.

\section{Methodology}\label{sec.methodology}

Instead of operating directly on optimization problem \eqref{eq.opt}, we reformulate it so that the budget constraint 
is expressed through a cost function defined over all subsets of nodes $2^V$, not just 
the feasible set $2^V_{\phi}$. 
We first show that the new optimization problem is equivalent to \eqref{eq.opt}. Unfortunately, it is hard to approximate 
without any dependency on the structure of the cost function. We then relax it to an optimization problem that uses a 
modular approximation to the cost function, but operates with reduced budget to satisfy the original budget constraint. 
This relaxation is the basis for our algorithms developed in the next section, and is sound if the cost function of 
the original problem has certain topological and analytical constraints. The following subsections explain our 
methodology in more details. 

\subsection{Expansion of the Cost Function}

We start by expanding the domain of cost function $c$ to power-set $2^V$, which will provide us with better 
insights on the computational complexity of the original problem. In particular, we consider the following 
expansion 
$c_e: 2^V \rightarrow \R$:
\begin{align}\label{eq_cost_e} 
c_e(S) = c(S_p) + \sum_{v \in S \backslash S_p} \min_{v_p \in V \cap \neighbor_v}  \E(\pay(v, v_p)),  
\end{align}
where $S_p$ is a set of all nodes in $S$ who also have a peer in $S$, i.e., 
$S_p = \{v \in S : \exists v_p \in S \text{ s.t. } v_p \in \neighbor_v \ \}$. 
In other words, cost function $c_e$ acts as if all the nodes in $S$ who have a peer in $S$ are rewarded as usual, while those that
do not have a peer in $S$ are rewarded with the expected payoff they obtain when scored against the worst peer. Notice that 
$c_e(S) = c(S)$ for all $S \in 2^V_{\phi}$, which makes the expansion sound. 
We denote by $c_e(v|S)$ the marginal increase of cost $c_e$ for adding an element $v$ to $S$, i.e., $c_e(v|S) = 
c_e(S \cup \{v\}) - c_e(S)$. We establish the monotonicity of cost function $c_e$ with the following lemma; notice, however, 
that the cost function is not necessarily sub/super-modular. 

\begin{lemma}\label{lm_monotone}
Cost function $c_e$ defined by \eqref{eq_cost_e} is monotone. Furthermore: $c_e(\{v\}) \le c_e(v|S)$, for all $S \in 2^V \backslash \{ v \}$.   
\end{lemma}

\subsection{Hardness Result}

To prove the complexity of our initial problem, we adapt optimization problem \eqref{eq.opt} to use the extended cost function $c_e$. 
In particular, we consider the optimization problem defined as: 
\begin{align}
\selectSet^* = \argmax_{\selectSet \in 2^\groundSet_\feasible \text{ s.t. } c_e(\selectSet) \leq \budget} \funcf(\selectSet) \label{eq.opt_e}
\end{align}
Clearly, any feasible solution to the problem \eqref{eq.opt_e} is also a feasible solution to the original problem due to the constraint 
$\selectSet \in 2^\groundSet_\feasible$, while the optimality alignment is ensured by having the same objective value.  
Now, to show the hardness result for approximating $\opt$, we reduce the maximum clique problem to optimization 
problem \eqref{eq.opt_e} in a computationally efficient way, thus, obtaining: 
  
\begin{theorem}\label{thm_inapproximability}
For any $\epsilon > 0$, it is NP-hard to find a solution $S$ to optimization problem \eqref{eq.opt_e} (and thus \eqref{eq.opt}) 
such that $ \frac{f(S)}{f(\opt)} \ge \frac{1}{|V|^{1-\epsilon}}$.
\end{theorem}
\begin{proof}
Consider an arbitrary undirected unweighted graph $G' = (V, E')$ for which we wish to compute the maximum clique. 
To reduce the maximum clique problem to \eqref{eq.opt_e}:
1) define function $f$ as $f(S) = |S|$, which is clearly monotone and submodular;    
2) define payment function as: $\tau(v, v_p) = \tau(v_p, v) = \pay_{max}$ if $(v, v_p) \notin E'$, and $\tau(v, v_p) = \tau(v_p, v) = \pay_{min}$ otherwise; 
3) set budget $\budget$ to $\budget = |V| \cdot \pay_{min}$;
4) and set $\pay_{max} > \budget$.
Notice that such an arrangement induces a fully connected graph $\graph$. Furthermore, we defined deterministic payment functions $\tau(v, v_p)$ and $\tau(v_p, v)$, 
but one can use $\E(\tau(v, v_p))$ and $\E(\tau(v_p, v))$ instead.
Points 2 and 4 ensure that any solution to optimization problem \eqref{eq.opt_e} is a clique in graph $G'$; otherwise, the budget constraints would be violated 
in solving \eqref{eq.opt_e}. 
Likewise, points 2 and 3 ensure that any clique is permitted as a potential solution w.r.t. the budget constraint. 
Finally, point 1 ensures that we search for a clique with the maximum number of vertices. Since the reduction is computationally 
efficient (polynomial in the graph size, i.e. $|V|$ and $\edges$), optimization \eqref{eq.opt} is at least as hard as the maximum clique problem.  Using the fact 
that the maximum clique problem is hard to approximate within factor $\frac{1}{|V|^{1-\epsilon}}$ \cite{Hastad1999}, we obtain the claim.
\end{proof}

\subsection{Structural Properties of the Cost Function} 
To cope with the computational hardness of the problem at hand, we identify two structural properties of cost function $c_e$ 
(or equivalently, the structural properties of payment function $\tau$). The first one is related to {\em topological} properties of graph 
$\graph$, and can be quantified with the maximum number of peers that a node in the graph can have. As explained earlier in this 
section, we denote this number by $\omega$. 

The second property is similar to the notion of {\em curvature} of a submodular function \cite{Iyer:2013:COA:2999792.2999918}, but now defined over cost 
function $c_e$ that is not necessarily sub/super-modular. In particular, we define the {\em slope} $\alpha$ of cost 
function $c_e$ as: 
\begin{align*}
\alpha = 1-\min_{v \in V, S \in 2^V \backslash \{ v \}} \frac{c_e(\{v\})}{c_e(v|S)}.
\end{align*}
The slope of cost function $c_e$, as defined above, measures how much 
marginal gains of $c_e$ change as we add more to initially empty set 
of selected nodes.\footnote{That is, $\alpha$ quantifies
the maximum increase in $c_e$ for adding a node $v$ (see Lemma \ref{lm_monotone}).} 
Intuitively, it measures the deviation of $c_e$ from modularity. 
A specific case of our interest is when $\alpha = 0$, which indicates that $c_m$ 
is modular and, thus, can be decomposed into a sum of costs $c_m$ dependent only on one vertex, 
i.e., $\sum_{v \in S} c_m(v)$. In the next subsection, we discuss how to utilize modular approximations of $c_e$ 
when $c_e$ itself is not modular. First, let us upper bound the slope $\alpha$ using the fact that payments are bounded. 

\begin{lemma}\label{lm_slope}
The slope of cost function $c_e$ is upper-bounded by $\alpha  \le 1- \frac{\pay_{min}}{\omega \cdot \pay_{max}}$.
\end{lemma}

\subsection{Relaxed Optimization Problem}

To make use of the structural constraints of cost function $c_e$, let us consider a relaxed version of optimization problem \eqref{eq.opt_e} with budget constraints defined via a {\em modular lower bound} to cost function $c_e$, denoted by $c_M$. As we show in the next section, for such a relaxation, one can develop a greedy approach that has provable approximation guarantees on the quality of the obtained solution relative to $\opt$. More precisely, consider the modular function 
$c_M:2^V \rightarrow \R$ defined via a cost function $c_m: V \rightarrow \R$:
\begin{align}\label{eq_modular_approx}
&c_M(S) = \sum_{v \in S} c_m(v) \\
\textnormal{ where } &c_m(v) = \min_{v_p \in \neighbor_v} \E(\pay(v, v_p)). \nonumber
\end{align} 
Clearly, $c_M(S)$ lower bounds $c_e(S)$ as it calculates the expected payoffs of nodes in $S$ when they are  scored against their worst peers (not necessarily in $S$). Now, we relax optimization problem \eqref{eq.opt_e} to:
\begin{align}
\selectSet^* = \argmax_{\selectSet \in 2^\groundSet_\feasible \text{ s.t. } c_M(\selectSet) \leq \budget'} \funcf(\selectSet) \label{eq.opt_2}.
\end{align}
In order to make the relaxation sound, any selected set $\selectSet^*$ in problem \eqref{eq.opt_2} should also be feasible in problem \eqref{eq.opt_e} (and thus \eqref{eq.opt}). 
We can ensure this by reducing the available budget, i.e., by making $B'$ appropriately smaller than $B$. Using the slope of cost $c_e$, we can obtain that the following budget 
reduction satisfies our requirement.
\begin{lemma}\label{lm_reduced_budget}
Any feasible solution $S$ to optimization problem $\eqref{eq.opt_2}$ is also a feasible solution to optimization problem \eqref{eq.opt_e} (and thus \eqref{eq.opt}) for $B' \le (1-\alpha) \cdot B$, where $\alpha$ is the slope of cost function $c_e$. 
\end{lemma}

\begin{algorithm}[t!]
\nl {\bf Input}:\\
	\begin{itemize}
 		\item[$\Square$] PPC graph: $\graph(\groundSet,\edges)$\; 
 		\item[$\Square$] Utility function : $\funcf$; budget $\budget$\;
		\item[$\Square$] Cost function : $c$, slope $\alpha$, modular approx. $c_M$\;
    \end{itemize}
\nl {\bf Output}: selected set $\selectSetAlg$\;
\nl {\bf Initialize}:\\
	\begin{itemize}
 		\item[$\Square$] $t=0$; $\selectSetAlg = \emptyset$; $\text{ budget } \budget^t = (1-\alpha) \cdot \budget$ \; 
 		\item[$\Square$] $\supergroundsetalg$ $\supergroundset = \emptyset$\;
    \end{itemize}
\nl //\texttt{Create $\supergroundsetalg$ $\supergroundset$}\\
\nl	\ForEach {$v \in \groundSet$}{
\nl		$\neighbor_v \gets \{u: \{v,u\} \in \edges\}$ \;
\nl		\ForEach {$u \in  \neighbor_v$}{
\nl			$z = \{v, u\}$; $\supergroundset = \supergroundset \cup \{z\}$ \;
		}
	}
\nl //\texttt{Compute $\selectSetAlg$}\\
\nl \While { $\budget^t > 0$ }{
\nl		$z^*_t = \argmax_{z \in \supergroundset, z \backslash \selectSetAlg \ne \emptyset,  c_M(z \backslash \selectSetAlg) \le \budget^t} \frac{f(z | \selectSetAlg)}{c_M(z \backslash \selectSetAlg)}$ \;
\nl		\If{$z^*_t = NULL$}{ {\bf break}\;		}
\nl 		$\budget^{t+1} = \budget^t -c_M(z^*_t \backslash \selectSetAlg)$ \;
\nl		$\selectSetAlg = \selectSetAlg \cup z^*_t$ \;
\nl		$t=t+1$ \;
	}
\nl {\bf Output}: $\selectSetAlg$\\
\caption{Algorithm \ppcgreedyalg}
\label{alg_pcpgreedyalg} 
\end{algorithm}

\section{Algorithm}\label{sec.mechanism}

We now present a new greedy algorithm for solving the optimization problem
with peer-prediction constraints (PPC), called PPCGreedy (Algorithm \ref{alg_pcpgreedyalg}). It is similar to standard greedy approaches for submodular maximization with budget constraints (e.g., \citet{2015-hcomp_nushi_crowd_access_path}), but it additionally ensures that a tentative output $\selectSetAlg$ at a certain iteration is an element of $2^\groundSet_\feasible$. To do so, it initially constructs a set of couples $\supergroundset $ that contains all the peer pairs and selects at each iteration $t$ either a node that already has a peer in the selected set or a pair of nodes that are peers. The selection procedure makes a choice $z^*$ that maximizes the ratio between the utility gain and the cost increase, while not exceeding a given budget $B' = (1-\alpha) \cdot B$. If there are multiple choices that maximize this ratio, the selection procedure selects one of them, whereas if there is no choice that fits the budgets constraints, $z^*$ is set to $NULL$, which ends the search and outputs the current solution $\selectSetAlg$. 

\subsection{Analysis}

We will now show the main property of our algorithm: its near optimality when cost function $c_e$ has a low slope $\alpha$,
i.e., when the difference between $\pay_{max}$ and $\pay_{min}$ is small.
Notice that parameters $\pay_{max}$ and $\pay_{min}$ are controllable through our design of a peer-prediction method $\pay$ and the requirements on minimal expected payments, which implies that $\alpha$ can be tuned. For all practical reasons, it is also reasonable to assume that $\frac{(1-\alpha) \cdot B}{\tau_{max}} > 2$, which simply states that our algorithm is always able to initially select any pair of nodes.

\begin{theorem}\label{thm_main}
Let the maximal relative difference between modular costs of two peer nodes be bounded by $r$, i.e., $r \ge \max_{v \in V, v_p \in \neighbor_v}\frac{c_m(v)}{c_m(v_p)}$, and let $\gamma = \max_{v \in V}\frac{c_{m}(v)}{B'} \in (0, \frac{1}{2})$. 
Then, the output $\selectSetAlg$ of Algorithm \ref{alg_pcpgreedyalg} has the following guarantees on the utility: 
\begin{align}
f(\selectSetAlg) \ge \left (1- e^{-\frac{(1-\alpha) \cdot (1-2 \cdot \gamma)}{1+r}} \right ) \cdot f(\opt).
\end{align} 
\end{theorem}
\begin{proof}[Proof (Sketch)] The proof of the theorem is non-trivial, so we outline only its basic steps (see \cite{RSKF:17} for more details). Using the fact 
that $f$ is submodular, while Algorithm \ref{alg_pcpgreedyalg} is greedy in terms of $f/c_M$ ratio, we show that:
\begin{align*}
f(z^*_t|S_t) \ge \frac{c_M(z^*_t\backslash S_t)}{(1+r) \cdot B} \cdot [f(\bar\opt) - f(S_t)],
\end{align*}
where $S_t$ is equal to $\selectSetAlg$ at time-step $t$, while $\bar\opt$ is the optimum solution to optimization problem $\eqref{eq.opt_2}$ when budget $B' = B$. Now, following the the proofs of related results for submodular maximization under budget constraints (e.g., \citet{Sviridenko:2004:NMS:2308897.2309164}, \citet{2015-hcomp_nushi_crowd_access_path}), and adapting them to our setting, we obtain that:
\begin{align*}
f(\selectSetAlg) &\ge \left (1- e^{-\frac{(1-\alpha) \cdot (1-2 \cdot \gamma)}{1+r}} \right ) f(\bar\opt).
\end{align*}
As we argue in the full proof, $f(\bar\opt) \ge f(\opt)$ because $\bar\opt$ is obtained for the same budget as $\opt$, but the cost $c_M$ that lower bounds $c$. Together with the above inequality, this implies the statement of the theorem. 
\end{proof}
We see that the quality of the approximation ratio depends on the structural properties of the cost function, including slope $\alpha$, the maximum cost discrepancy between 
two nodes measured by $r$, and the maximum fraction of the budget assigned to a node, measured by $\gamma$. As $\alpha$ approaches its maximum value, i.e., $\alpha \rightarrow 1$, the approximation ration goes to $0$. This is consistent with the hardness result presented in Section 'Methodology', which shows the necessity of imposing structural constraints. One can reach a similar conclusion by analyzing $r$ as it goes to its maximal value, i.e., $r \rightarrow \infty$.  

To see this more clearly, we can express the results of the theorem in terms of the original optimization problem and the structural properties of payment function $\tau$.
Using the bound on slope $\alpha$ (Lemma \ref{lm_slope}), the boundedness of payments, which imply $r \le \frac{\pay_{max}}{\pay_{min}}$,
we obtain:   
\begin{corollary}\label{cor_main_result}
Assuming $B > 2 \cdot \frac{\omega \cdot \pay_{max}^2}{\pay_{min}}$, the output $\selectSetAlg$ of Algorithm \ref{alg_pcpgreedyalg} has the following guarantees on the utility: 
\begin{align}
f(\selectSetAlg) \ge \left (1- e^{-\frac{\pay_{min}^2}{\omega \pay_{max}^2 } \cdot \left (\frac{1}{2} - \frac{\omega \cdot \pay_{max}^2}{B \cdot \pay_{min}} \right)} \right ) \cdot f(\opt).
\end{align} 
\end{corollary}
Therefore, whenever the maximum payment $\tau_{max}$ or the number of possible peers $\omega$ go to large values, the approximation factor becomes negligible. Notice that the number of possible peers $\omega$ is dependent on $\tau_{min}$, so we can alternatively say that for small values of $\tau_{min}$, i.e., $\tau_{min} \approx 0$, the quality of the obtained greedy solution is relatively low. In practice, however, we can often avoid these corner cases by adjusting the payment function, and thus $\tau_{max}$ and $\tau_{min}$.

\subsection{More Efficient Budget Expenditure}

The PPCGreedy algorithm, as described by Algorithm \ref{alg_pcpgreedyalg} does not necessarily spend the full budget on incentivizing nodes. This is because we use a reduced budget $B'$ when running the main steps of the algorithm. One can achieve a better budget efficiency by iteratively calling PPCGreedy method, as shown in Algorithm \ref{alg_pcpgreedyalgrec}, that we refer to as PPCGreedyIter. It is important to note that in the sub-procedure $PPCGreedy$ we take into account the current set of selected nodes $\selectSetAlg$ when examining the feasibility of a solution and evaluating the utility and cost functions. The budget reduction in the $PPCGreedy$ subroutine can, on the other hand, be done with the same (initial) $\alpha$. The procedure terminates when no new node is added, which is equivalent to the budget not changing between two consecutive iterations.  

The utility function $f$ is always evaluated with the selected set of nodes $\selectSetAlg$ from previous iterations, in the algorithm denoted by $f(\cdot \cup \selectSetAlg)$. The same is true for cost function $c$, denoted by $c(\cdot \cup \selectSetAlg)$, and its modular approximation $c_M$. Due to monotonicity of $f$, this means that the reached solution is always as good as the one obtained by PPCGreedy. Furthermore, the cost of the solution is within the budget constraints: this is because $c_e(\{ v \}) \le c_e(v | S)$  (Lemma \ref{lm_monotone}), so the slope $\alpha$ defined on $c_e(\cdot)$ upper bounds the one defined on $c(\cdot \cup \selectSetAlg)$, which implies that the subroutine $PPCGreedy$ makes a proper budget reduction. Therefore, the results of Theorem \ref{thm_main} and Corollary \ref{cor_main_result} are preserved. 

\begin{algorithm}[t!] 
\nl {\bf Output}: selected set $\selectSetAlg$\;
\nl {\bf Initialize}:\\
	\begin{itemize}
 		\item[$\Square$] $t=1$; $\selectSetAlg = \emptyset$; $\text{ budget } \budget^0 = \budget + \epsilon$; $\budget^1 = \budget$ \; 
    \end{itemize}
\nl //\texttt{Compute $\selectSetAlg$}\\
\nl \While { $\budget^t < \budget^{t-1}$ }{
\nl		$\selectSetAlg_{t} = PPCGreedy(f(\cdot \cup \selectSetAlg), c( \cdot \cup \selectSetAlg), \budget^t,\selectSetAlg)$ \;
\nl 		$\budget^{t+1} = \budget - c(\selectSetAlg_{t} \cup \selectSetAlg)$ \;
\nl		$\selectSetAlg = \selectSetAlg \cup \selectSetAlg_{t}$ \;
\nl		$t=t+1$ \;
	}
\nl {\bf Output}: $\selectSetAlg$\\
\caption{Algorithm \ppcgreedyalgiter}
\label{alg_pcpgreedyalgrec} 
\end{algorithm}
  
\vspace{-2mm}
\section{Experimental Evaluation}\label{sec.experiments}
\vspace{-1mm}

To evaluate our approach, we use a crowd  sensing test-bed of \citet{AS:17}, constructed from real measurements 
of $CO_2$ and user locations across an urban area. The concentrations of $CO_2$ in the city of Zurich were acquired with a 
{\em NODE+}\footnote{http://www.variableinc.com/node1} sensor. These measurements were used to fit a Gaussian variogram 
whose parameters indicate that 
the relevant correlation range between two measurement locations is about $R = 236$ meters. 
We use this distance to define a {\em disk} 
coverage function---for a set of points of interest, we count how many of these are within $R$ meters away from the set of selected 
points. More formally, given a set of points $S$ that represent the location of the selected sensors and set of points $S_{poi}$ that 
represent locations for which we would like to obtain $CO_2$ measurements, the objective function $f$ is defined as:
$f(S) = \sum_{s \in S_{poi}} \mathds 1_{min_{s' \in S} d(s, s') \le R}$. Here,
$\mathds 1_{cond}$ is an indicator variable, evaluating to $1$ when $cond$ is satisfied, and is $0$ otherwise, while $d(s,s')$ measures 
the distance in meters between locations $s$ and $s'$. The function $f$ is a coverage function, which is monotone and 
submodular \cite{2012-survey_krause_submodular}.

Points of interest $S_{poi}$ are predefined, and in total, there are $300$ of them. These were obtained using a publicly
available data ({\em OpenStreetMap}\footnote{http://wiki.openstreetmap.org/wiki/Node}), from which we randomly selected 
$300$ locations from an area in the center of New York City. To identify the locations of available crowd-sensors, 
i.e., the ground set $V$, we use the population statistics of the test-bed, which give us the likelihood of a user appearing 
in one of the $300$ points. This statistics is inferred from a publicly accessible
dataset ({\em Strava}\footnote{http://metro.strava.com/}) that contains the mobility patterns of cyclists for a period of $6$ days. 
We sample from the likelihood $1000$ points to obtain sensing locations and then we perturb them by $50$ meters. 

As a peer prediction scoring rule, we use the output agreement mechanism as described in Section {\em Problem Statement}. The expected 
score of this mechanism for two points $s$ and $s'$ in truthful reporting regime is equal to $\E(d(s,s')^2)$. We approximate
the expected value of OA for two sensors $s$ and $s'$ by using the variogram of the test-bed. 
More precisely,
given the range parameter $R = 236$, we estimate the expected payoff between two sensors as:
$\E(\tau(s_1, s_2)) = 1-e^{-\frac{d(s_1,s_2)^2}{a\cdot R^2}}$,
where $a$ is set to $\frac{1}{3}$.

\begin{figure*}
\begin{multicols}{2}
    \centering
    \includegraphics[width=0.83\linewidth]{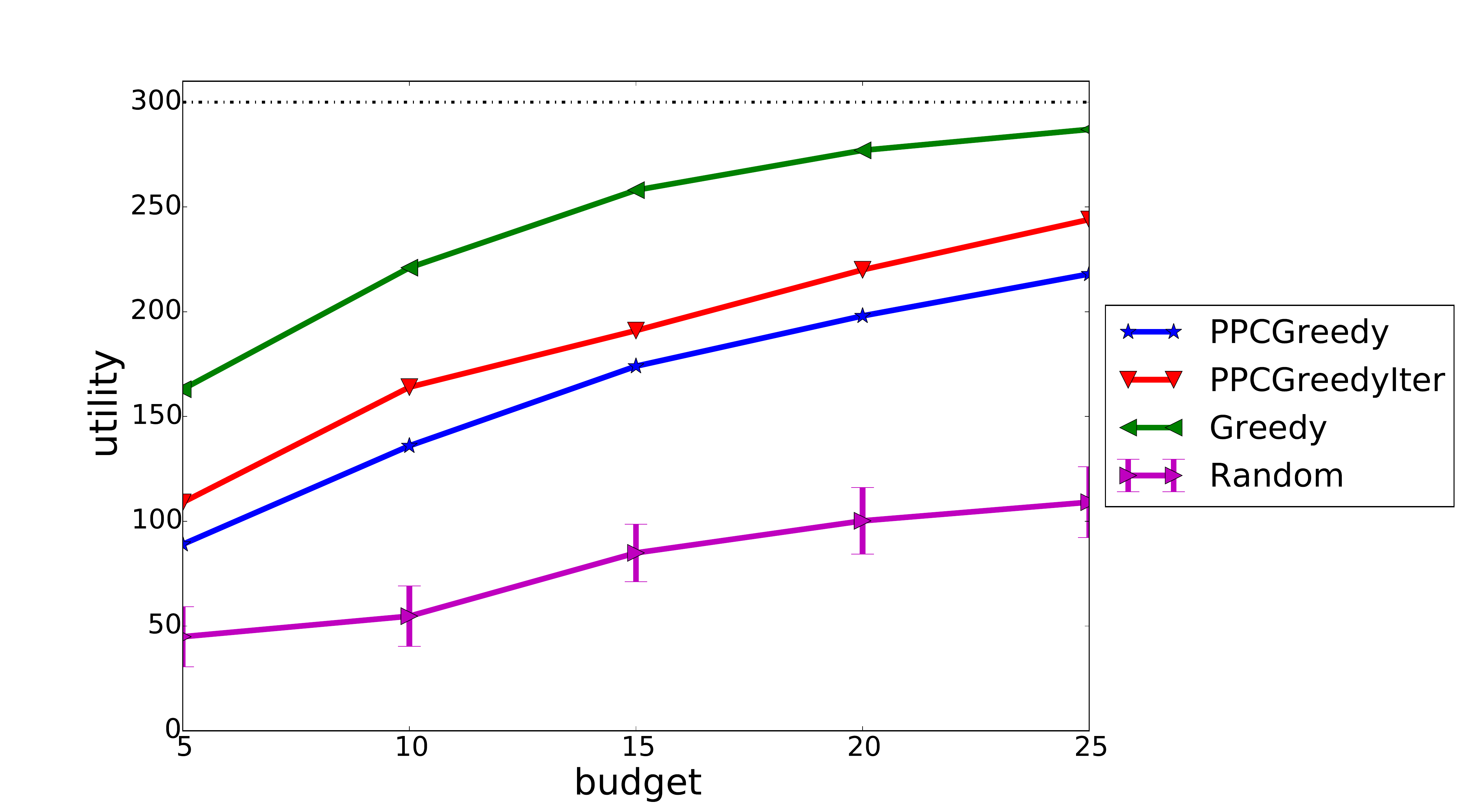}
    \subcaption{Performance as we vary the available budget.}\label{fig:budget}
    \par 
    \centering
    \includegraphics[width=0.83\linewidth]{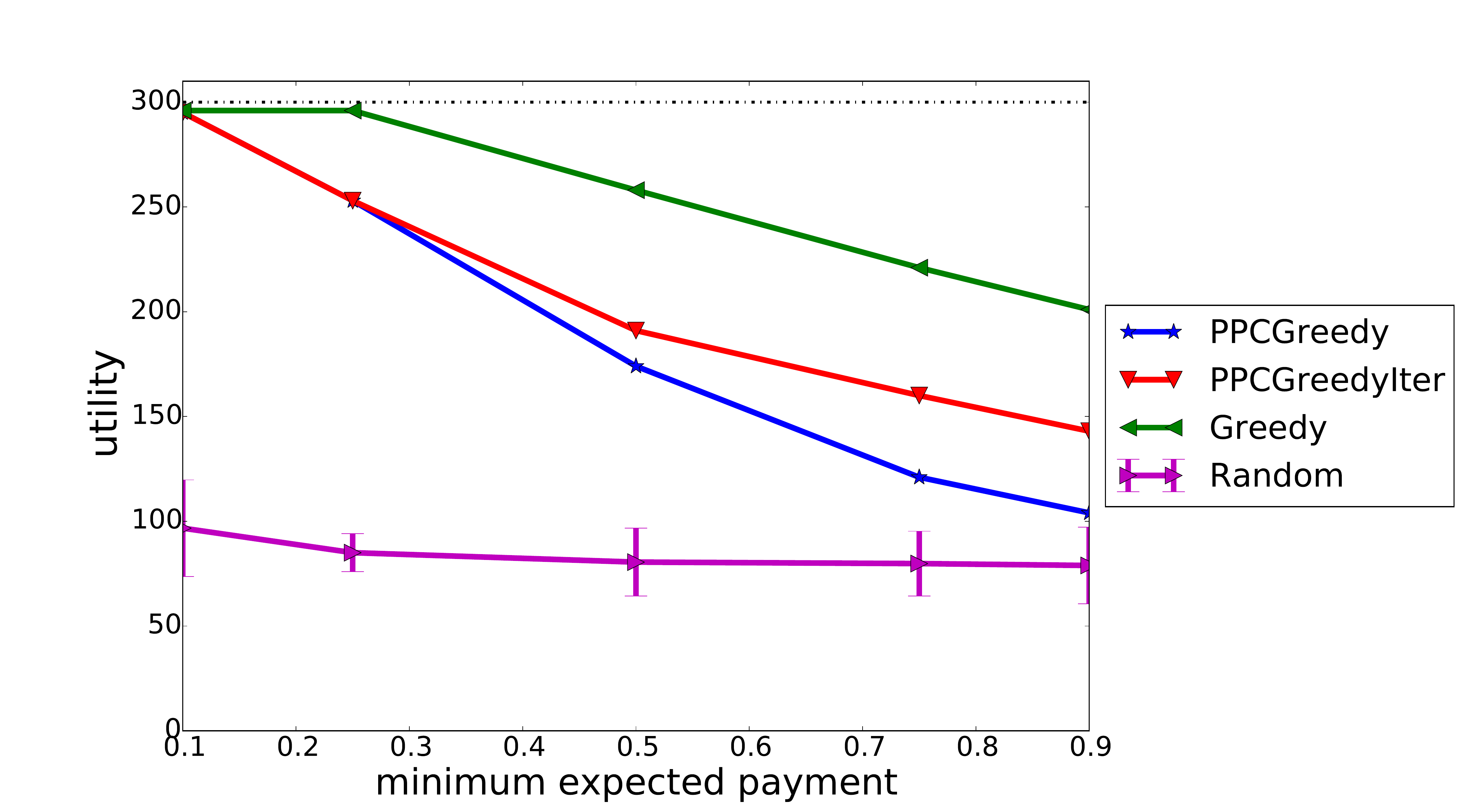}
    \subcaption{Performance as we vary the minimum payments.}\label{fig:min_pay}
    \par 
    \end{multicols}
\vspace{-5mm}
\caption{Experimental results show how the utility $f(\selectSetAlg)$ changes as we increase budget $B$ or the minimum expected payments $\pay_{min}$. $\selectSetAlg$ is the output of different methods. Increasing $B$ is beneficial as it allows more sensors to be selected. On the other hand, as we increase $\pay_{min}$, the sensors have less peers and they need to be payed more.  {\em Random} is run $10$ times and we show the means and standard deviations. }
\vspace{-0.5mm}
\end{figure*}

\textbf{Results.} We test our approaches, PPCGreedy and PPCGreedyIter, against two other baselines: (a) a random selection (denoted by {\em Random}) that satisfies peer-prediction constraints, 
(b) a greedy approach (denoted by {\em Greedy}) that assumes it suffices to reward each sensor with $\tau_{min}$, without providing incentives for accurate reporting. 
Clearly, the latter baseline represents an optimistic approach whose performance upper bounds that of the proposed algorithms, while the former one is 
likely to lower bound their performance. In all the cases, the expected budget is at most $B$. 

We perform two different tests. In the first test, we vary the total available budget $B$ from $5$ to $25$ at steps of size $5$. At the same 
time, we keep the minimal expected payment to $\tau_{min} = 0.5$. As we can see from Figure \ref{fig:budget}, as the total budget increases, all the 
methods perform better. However, the increase is more notable for non-random algorithms. The performance of PPCGreedyIter is generally 
better than the one of PPCGreedy, and this is due to the spent budget --- as explained in the previous section, PPCGreedy is only the first step of 
PPCGreedyIter, which further iteratively runs PPCGreedy on the remaining budget. 

In the second test, we vary the minimal expected payment $\tau_{min}$, which now takes values in $\{0.1, 0.25, 0.5, 0.75, 0.9\}$. Budget $B$ is set to $15$. The results are given in Figure \ref{fig:min_pay}. 
Except for Random, 
a general trend is that increasing the minimal payments leads to lower performance, which is not surprising given that the number of peers  and the budget per sensor
decrease in that case. Random in general performs much worse than other techniques for all values of $\tau_{min}$. Moreover, notice that the discrepancy in performance between
PPCGreedy, PPCGreedyIter, and Greedy, increases with $\tau_{min}$. Initially, all the non-random algorithms find an optimum of function $f$, achieving the utility 
equal to $300$. 

\vspace{-2mm}
\section{Related Work}\label{sec.related}

{\bf Information elicitation.} A standard incentives for quality are typically categorized into 
gold standard techniques, such as proper scoring rules \cite{GR:07} or prediction markets \cite{CP:07}, or 
peer-prediction techniques, such as the classical peer-prediction \cite{MRZ:05}  or Bayesian truth serums \cite{P:04,WP:12a,RF:13}.
We focus in this paper on peer-prediction techniques due to the scalability of elicitation without verification 
for acquiring large amounts of highly distributed information. 
Recently, several peer-predictions were proposed for various crowdsourcing scenarios. 
These included micro-task crowdsourcing \cite{DG:13},  opinion polling \cite{JF:11}, information markets \cite{A:17}, peer grading \cite{Shnayder:2016:EC}, 
and most importantly for this work, crowdsensing \cite{RF:16b}. The proposed mechanisms for these domains follow the 
standard principles of the classical peer-prediction, e.g., incentivizing participants by comparing their reports and 
placing higher scores for a priori less likely matches. However, they also often extend the design of 
the original methods by making them more robust in terms of the required number of participants and the knowledge about  
them, the heterogeneity of users and tasks, or susceptibility to collusive behaviors \cite{FR:17}. 
We analyze orthogonal characteristics important for deploying such mechanisms in practice, i.e., 
budget and cost acquisition constraints. Although the prior work (e.g., \citet{DBLP:conf/ijcai/LiuC16}) does study meta-mechanisms 
that make peer-predictions proper in terms of effort exertion, it is often based on scaling techniques,
which either ignore budget limitations or the cost of effort.

{\bf Submodular function maximization.} From the technical side, the most important aspect relates to submodular function maximization.
While there is a sizeable literature on this topic (e.g., \citet{krause2011submodularity}, \citet{2012-survey_krause_submodular}), we 
mostly focus on the prior work that is closely related to the techniques used in this paper. Our basic objective is a subset selection 
under budget (knapsack) constraints (e.g., \citet{Sviridenko:2004:NMS:2308897.2309164}), and we base our algorithmic techniques on a 
simple greedy approach \cite{2015-hcomp_nushi_crowd_access_path}. Notice that we additionally have a graph based constraint,
which is in spirit similar to \citet{singla15netexp}, although we are solving a different optimization problem. Arguably,
this paper is most related to submodular maximization with submodular budget constraints \cite{Iyer:2013:SOS:2999792.2999884}; 
contrary to this work, our budget constraints are not necessarily sub/super-modular. It is also worth mentioning the hardness results that relate to the ones 
obtained in this paper, such as  the inapproximability  of the maximum of a submodular non-monotone, possibly negative, profit 
function \cite{Feige:2008:CAM:1367497.1367521}.  

\vspace{-2mm}
\section{Conclusion}\label{sec.conclusions}

In this paper, we have introduced an information elicitation model for data collection from distributed 
sources when the incentive mechanism is based on peer-predictions. We have shown that optimal 
information gathering is computationally infeasible in that even approximating the optimal solution is 
NP-hard. However, given structural constraints on peer-prediction incentives, we have 
proposed two greedy methods that achieve good performance relative to the optimum, and 
have tested their performance empirically on a realistic crowd-sensing test-best.


\vspace{0.5mm}
{\bfseries Acknowledgments}
This work was supported in part by the Swiss National Science Foundation, Nano-Tera.ch program as part of the Opensense II project, ERC StG 307036, a SNSF Early Postdoc Mobility fellowship, and a Facebook Graduate fellowship.

\bibliographystyle{aaai}
\bibliography{Radanovic-Singla}

\begin{thebibliography}{}

\bibitem[\protect\citeauthoryear{Baillon}{2017}]{A:17}
Baillon, A.
\newblock 2017.
\newblock Bayesian markets to elicit private information.
\newblock {\em Proceedings of the National Academy of Sciences}
  114(30):7958--7962.

\bibitem[\protect\citeauthoryear{Chen and Pennock}{2007}]{CP:07}
Chen, Y., and Pennock, D.~M.
\newblock 2007.
\newblock A utility framework for bounded-loss market makers.
\newblock In {\em Proceedings of the Twenty-Third Conference on Uncertainty in
  Artificial Intelligence}.

\bibitem[\protect\citeauthoryear{Dasgupta and Ghosh}{2013}]{DG:13}
Dasgupta, A., and Ghosh, A.
\newblock 2013.
\newblock Crowdsourced judgement elicitation with endogenous proficiency.
\newblock In {\em Proceedings of the 22nd ACM International World Wide Web
  Conference}.

\bibitem[\protect\citeauthoryear{Faltings and Radanovic}{2017}]{FR:17}
Faltings, B., and Radanovic, G.
\newblock 2017.
\newblock {\em Game Theory for Data Science: Eliciting Truthful Information}.
\newblock Morgan \& Claypool Publishers.

\bibitem[\protect\citeauthoryear{Faltings, Li, and Jurca}{2014}]{faltings2013}
Faltings, B.; Li, J.~J.; and Jurca, R.
\newblock 2014.
\newblock {I}ncentive {M}echanisms for {C}ommunity {S}ensing.
\newblock {\em IEEE Transaction on Computers} 63(1):115--128.

\bibitem[\protect\citeauthoryear{Feige \bgroup et al\mbox.\egroup
  }{2008}]{Feige:2008:CAM:1367497.1367521}
Feige, U.; Immorlica, N.; Mirrokni, V.; and Nazerzadeh, H.
\newblock 2008.
\newblock A combinatorial allocation mechanism with penalties for banner
  advertising.
\newblock In {\em Proceedings of the 17th International Conference on World
  Wide Web}.

\bibitem[\protect\citeauthoryear{Frongillo and Witkowski}{2016}]{FW:16}
Frongillo, R., and Witkowski, J.
\newblock 2016.
\newblock A geometric method to construct minimal peer prediction mechanisms.
\newblock In {\em Proceedings of the 30th AAAI Conference on AI}.

\bibitem[\protect\citeauthoryear{Gao, Wright, and
  Leyton{-}Brown}{2016}]{DBLP:journals/corr/GaoWL16}
Gao, A.; Wright, J.~R.; and Leyton{-}Brown, K.
\newblock 2016.
\newblock Incentivizing evaluation via limited access to ground truth:
  Peer-prediction makes things worse.
\newblock {\em CoRR} abs/1606.07042.

\bibitem[\protect\citeauthoryear{Gneiting and Raftery}{2007}]{GR:07}
Gneiting, T., and Raftery, A.~E.
\newblock 2007.
\newblock Strictly proper scoring rules, prediction, and estimation.
\newblock {\em Journal of the American Statistical Association} 102:359--378.

\bibitem[\protect\citeauthoryear{Hastad}{1999}]{Hastad1999}
Hastad, J.
\newblock 1999.
\newblock Clique is hard to approximate within $n^{1-\epsilon}$.
\newblock {\em Acta Math.} 182(1):105--142.

\bibitem[\protect\citeauthoryear{Iyer and
  Bilmes}{2013}]{Iyer:2013:SOS:2999792.2999884}
Iyer, R., and Bilmes, J.
\newblock 2013.
\newblock Submodular optimization with submodular cover and submodular knapsack
  constraints.
\newblock In {\em Proceedings of the 26th International Conference on Neural
  Information Processing Systems}.

\bibitem[\protect\citeauthoryear{Iyer, Jegelka, and
  Bilmes}{2013}]{Iyer:2013:COA:2999792.2999918}
Iyer, R.; Jegelka, S.; and Bilmes, J.
\newblock 2013.
\newblock Curvature and optimal algorithms for learning and minimizing
  submodular functions.
\newblock In {\em Proceedings of the 26th International Conference on Neural
  Information Processing Systems}.

\bibitem[\protect\citeauthoryear{Jurca and Faltings}{2011}]{JF:11}
Jurca, R., and Faltings, B.
\newblock 2011.
\newblock Incentives for answering hypothetical questions.
\newblock In {\em Workshop on Social Computing and User Generated Content}.

\bibitem[\protect\citeauthoryear{Kong and Schoenebeck}{2016}]{KS:16b}
Kong, Y., and Schoenebeck, G.
\newblock 2016.
\newblock A framework for designing information elicitation mechanisms that
  reward truth-telling.
\newblock {\em CoRR} abs/1603.07751.

\bibitem[\protect\citeauthoryear{Krause and
  Golovin}{2012}]{2012-survey_krause_submodular}
Krause, A., and Golovin, D.
\newblock 2012.
\newblock Submodular function maximization.
\newblock {\em Tractability: Practical Approaches to Hard Problems} 3:19.

\bibitem[\protect\citeauthoryear{Krause and
  Guestrin}{2011}]{krause2011submodularity}
Krause, A., and Guestrin, C.
\newblock 2011.
\newblock Submodularity and its applications in optimized information
  gathering.
\newblock {\em ACM Transactions on Intelligent Systems and Technology} 2(4):32.

\bibitem[\protect\citeauthoryear{Liu and Chen}{2016}]{DBLP:conf/ijcai/LiuC16}
Liu, Y., and Chen, Y.
\newblock 2016.
\newblock Learning to incentivize: Eliciting effort via output agreement.
\newblock In {\em Proceedings of the 25th International Joint Conference on
  Artificial Intelligence}.

\bibitem[\protect\citeauthoryear{Miller, Resnick, and
  Zeckhauser}{2005}]{MRZ:05}
Miller, N.; Resnick, P.; and Zeckhauser, R.
\newblock 2005.
\newblock Eliciting informative feedback: The peer-prediction method.
\newblock {\em Management Science} 51:1359--1373.

\bibitem[\protect\citeauthoryear{Nushi \bgroup et al\mbox.\egroup
  }{2015}]{2015-hcomp_nushi_crowd_access_path}
Nushi, B.; Singla, A.; Gruenheid, A.; Zamanian, E.; Krause, A.; and Kossmann,
  D.
\newblock 2015.
\newblock Crowd access path optimization: Diversity matters.
\newblock In {\em AAAI Conference on Human Computation and Crowdsourcing}.

\bibitem[\protect\citeauthoryear{Prelec}{2004}]{P:04}
Prelec, D.
\newblock 2004.
\newblock A bayesian truth serum for subjective data.
\newblock {\em Science} 34(5695):462--466.

\bibitem[\protect\citeauthoryear{Radanovic and Faltings}{2013}]{RF:13}
Radanovic, G., and Faltings, B.
\newblock 2013.
\newblock A robust bayesian truth serum for non-binary signals.
\newblock In {\em Proceedings of the 27th AAAI Conference on Artificial
  Intelligence}.

\bibitem[\protect\citeauthoryear{Radanovic \bgroup et al\mbox.\egroup
  }{2017}]{RSKF:17}
Radanovic, G.; Singla, A.; Krause, A.; and Faltings, B.
\newblock 2017.
\newblock Information gathering with peers: Submodular optimization with
  peer-prediction constraints (extended version).
\newblock {\em CoRR} abs/1711.06740.

\bibitem[\protect\citeauthoryear{Radanovic, Faltings, and Jurca}{2016}]{RF:16b}
Radanovic, G.; Faltings, B.; and Jurca, R.
\newblock 2016.
\newblock Incentives for effort in crowdsourcing using the peer truth serum.
\newblock {\em ACM Transactions on Intelligent Systems and Technology}
  7:48:1--48:28.

\bibitem[\protect\citeauthoryear{Shnayder \bgroup et al\mbox.\egroup
  }{2016}]{Shnayder:2016:EC}
Shnayder, V.; Agarwal, A.; Frongillo, R.; and Parkes, D.~C.
\newblock 2016.
\newblock Informed truthfulness in multi-task peer prediction.
\newblock In {\em Proceedings of the 2016 ACM Conference on Economics and
  Computation}.

\bibitem[\protect\citeauthoryear{Singla and Krause}{2013}]{singla13incentives}
Singla, A., and Krause, A.
\newblock 2013.
\newblock Incentives for privacy tradeoff in community sensing.
\newblock In {\em AAAI Conference on Human Computation and Crowdsourcing
  (HCOMP)}.

\bibitem[\protect\citeauthoryear{Singla \bgroup et al\mbox.\egroup
  }{2014}]{singla14stochasticPrivacy}
Singla, A.; Horvitz, E.; Kamar, E.; and White, R.~W.
\newblock 2014.
\newblock Stochastic privacy.
\newblock In {\em Proc. Conference on Artificial Intelligence (AAAI)}.

\bibitem[\protect\citeauthoryear{Singla \bgroup et al\mbox.\egroup
  }{2015}]{singla15netexp}
Singla, A.; Horvitz, E.; Kohli, P.; White, R.; and Krause, A.
\newblock 2015.
\newblock Information gathering in networks via active exploration.
\newblock In {\em IJCAI}.

\bibitem[\protect\citeauthoryear{Singla}{2017}]{AS:17}
Singla, A.
\newblock 2017.
\newblock {\em Learning and Incentives in Crowd-Powered Systems}.
\newblock Ph.D. Dissertation, ETH.

\bibitem[\protect\citeauthoryear{Sviridenko}{2004}]{Sviridenko:2004:NMS:2308897.2309164}
Sviridenko, M.
\newblock 2004.
\newblock A note on maximizing a submodular set function subject to a knapsack
  constraint.
\newblock {\em Oper. Res. Lett.} 32(1):41--43.

\bibitem[\protect\citeauthoryear{Tschiatschek, Singla, and
  Krause}{2017}]{tschiatschek17ordered}
Tschiatschek, S.; Singla, A.; and Krause, A.
\newblock 2017.
\newblock Selecting sequences of items via submodular maximization.
\newblock In {\em Proceedings of the 31th AAAI Conference on Artificial
  Intelligence (AAAI'17)}.

\bibitem[\protect\citeauthoryear{von Ahn and
  Dabbish}{2004}]{vonAhn:2004:LIC:985692.985733}
von Ahn, L., and Dabbish, L.
\newblock 2004.
\newblock Labeling images with a computer game.
\newblock In {\em Proceedings of the SIGCHI Conference on Human Factors in
  Computing Systems}.

\bibitem[\protect\citeauthoryear{Waggoner and
  Chen}{2014}]{DBLP:conf/hcomp/WaggonerC14}
Waggoner, B., and Chen, Y.
\newblock 2014.
\newblock Output agreement mechanisms and common knowledge.
\newblock In {\em Proceedings of the Second {AAAI} Conference on Human
  Computation and Crowdsourcing}.

\bibitem[\protect\citeauthoryear{Witkowski and Parkes}{2012}]{WP:12a}
Witkowski, J., and Parkes, D.~C.
\newblock 2012.
\newblock A robust bayesian truth serum for small populations.
\newblock In {\em Proceedings of the 26th AAAI Conference on Artificial
  Intelligence}.

\bibitem[\protect\citeauthoryear{Witkowski \bgroup et al\mbox.\egroup
  }{2017}]{DBLP:conf/aaai/WitkowskiAUK17}
Witkowski, J.; Atanasov, P.; Ungar, L.~H.; and Krause, A.
\newblock 2017.
\newblock Proper proxy scoring rules.
\newblock In {\em Proceedings of the Thirty-First {AAAI} Conference on
  Artificial Intelligence}.

\end{thebibliography}

\clearpage
\appendix 
{\allowdisplaybreaks

\section*{Appendix: Information Gathering with Peers}

\subsection{Proof of Lemma 1} 
{\bf Statement: } {\em Cost function $c_e$ defined by \eqref{eq_cost_e} is monotone. Furthermore: $c_e(\{v\}) \le c_e(v|S)$, for all $S \in 2^V \backslash \{ v \}$.  }
\begin{proof}
Denote by $S_p^v$ a set of nodes in $V$ who have peers in $S \cup \{v\}$. 
Notice that  $S_p \subseteq S_p^v$. For a given set $S \in 2^V$ and $v \notin S$, 
we have: 
\begin{align*}
&c_e(S \cup \{ v \}) =  c(S_p^v) \\
&+ \sum_{v' \in (S \cup \{ v \}) \backslash S_p^v}  \min_{v_p \in V \cap \neighbor_{v'}} \E(\pay(v', v_p)) \\ 
&= \sum_{v' \in (S \cup \{ v \}) \cap S_p^v} \max_{v_p \in (S \cup \{v\}) \cap \neighbor_{v'}} \E(\pay(v', v_p)) \\
&+\sum_{v' \in (S \cup \{ v \}) \backslash S_p^v}  \min_{v_p \in V \cap \neighbor_{v'}} \E(\pay(v', v_p))\\
&\ge \sum_{v' \in S \cap S_p} \max_{v_p \in S \cap \neighbor_{v'}} \E(\pay(v', v_p)) \\
&+ \sum_{v' \in (S \cap S_p^v) \backslash S_p} \min_{v_p \in V \cap \neighbor_{v'}} \E( \pay(v', v_p))\\
&+\min_{v_p \in V \cap \neighbor_v} \E(\pay(v, v_p)) \\ 
&+\sum_{v' \in S \backslash S_p^v}  \min_{v_p \in V \cap \neighbor_{v'}} \E(\pay(v', v_p))\\
&= \sum_{v' \in S \cap S_p} \max_{v_p \in S \cap \neighbor_{v'}} \E(\pay(v', v_p)) \\
&+ \sum_{v' \in S \backslash S_p} \min_{v_p \in V \cap \neighbor_{v'}}  \E(\pay(v', v_p))\\
&+\min_{v_p \in V \cap \neighbor_v} \E(\pay(v, v_p)) \\
&\ge c(S_p) + \sum_{v' \in S \backslash S_p} \min_{v_p \in V \cap \neighbor_{v'}}  \E(\pay(v', v_p))\\
&+ \pay_{min} = c_e(S) + \pay_{min} > c_e(S).
\end{align*}
Furthermore, from the proof we see that $c_e({v}|S) \ge \min_{v_p \in V \cap \neighbor_v} \pay(v, v_p) = c_e({v})$,
thus proving the second part of the statement. 
\end{proof}

\subsection{Proof of Theorem 1}
{\bf Statement: } {\em For any $\epsilon > 0$, it is NP-hard to find a solution $S$ to optimization problem \eqref{eq.opt_e} (and thus \eqref{eq.opt}) 
such that $ \frac{f(S)}{f(\opt)} \ge \frac{1}{|V|^{1-\epsilon}}$.}
\begin{proof}
We prove the statement by reducing the maximum clique problem to optimization \eqref{eq.opt_e}. 
Consider an arbitrary undirected unweighted graph $G' = (V, E')$ for which we wish to compute the maximum clique. 
To reduce the maximum clique problem to \eqref{eq.opt_e}:
\begin{enumerate}
\item define function $f$ as $f(S) = |S|$, which is clearly monotone and submodular;    
\item define payment function as: $\tau(v, v_p) = \tau(v_p, v) = \pay_{max}$ if $(v, v_p) \notin E'$, and $\tau(v, v_p) = \tau(v_p, v) = \pay_{min}$ otherwise; 
\item set budget $\budget$ to $\budget = |V| \cdot \pay_{min}$;
\item and set $\pay_{max} > \budget$.
\end{enumerate}
Notice that such an arrangement induces a fully connected graph $\graph$. Furthermore, we defined deterministic payment functions $\tau(v, v_p)$ and $\tau(v_p, v)$, 
but one can use $\E(\tau(v, v_p))$ and $\E(\tau(v_p, v))$ instead. 
Points 2 and 4 ensure that any solution to optimization problem \eqref{eq.opt_e} is a clique in graph $G'$; otherwise, the budget constraints would be violated 
in solving \eqref{eq.opt_e}. 
Likewise, points 2 and 3 ensure that any clique is permitted as a potential solution w.r.t. the budget constraint. 
Finally, point 1 ensures that we search for a clique with the maximum number of vertices. Since the reduction is computationally 
efficient (polynomial in the graph size, i.e. $|V|$ and $\edges$), optimization \eqref{eq.opt} is at least as hard as the maximum clique problem.  Using the fact 
that the maximum clique problem is hard to approximate within factor $\frac{1}{|V|^{1-\epsilon}}$ \cite{Hastad1999}, we obtain the claim.
\end{proof}

\subsection{Proof of Lemma 2}
{\bf Statement: } {\em The slope of cost function $c_e$ is upper-bounded by:
\begin{align*}
\alpha  \le 1- \frac{\pay_{min}}{\omega \cdot \pay_{max}}.
\end{align*}
\begin{proof}
Notice that the slope is maximized when there exist a node $v$ such that:
\begin{itemize}
\item the expected payoff of $v$ when it's scored against its worst peer $v_{p, min}$ is $\pay_{min}$,
\item the expected payoff of $v$ increases to $\pay_{max}$ whenever any other peer $v_{p'} \ne v_{p, min}$ is used for scoring,
\item any $v$'s peer $v_{p'}$ (including $v_{p, min}$\footnote{The analysis slightly changes if $v_p$ is required to have the same payoff as $v$ when they
are mutual peers, but the main result stays the same.}) achieves expected payoff $\pay_{max}$ when scored against $v$ and otherwise, when
scored against some other peer, they achieve $\pay_{min}$.
\end{itemize} 
This gives us:
\begin{align*}
\alpha &\le 1 - \frac{\pay_{min}}{\pay_{max} + (\omega - 1) \cdot \pay_{max} - (\omega - 1) \cdot \pay_{min} }  \\
&\le 1 - \frac{\pay_{min}}{\omega  \cdot \pay_{max} }
\end{align*} 
\end{proof} 

\subsection{Proof of Lemma 3}
{\bf Statement: } {\em
Any feasible solution $S$ to optimization problem $\eqref{eq.opt_2}$ is also a feasible solution to optimization problem \eqref{eq.opt_e} (and thus \eqref{eq.opt}) for $B' \le (1-\alpha) \cdot B$, where $\alpha$ is the slope of cost function $c_e$. 
}
\begin{proof}
Since both problems require that $S \in 2^\groundSet_\feasible$, we only need to show that the budget constraint in $\eqref{eq.opt_e}$ is not violated when $S$ is selected. 
Let us enumerate the elements of $S$, so that each element is assigned an index $i$.
We have:
\begin{align*}
c_M(S) &= \sum_{v \in S} c_m(v) = \sum_{v \in S} c_e(\{v\}) \\
&= \sum_{i = 1}^{|S|} \frac{c_e(\{v_i\})}{c_e(\{v_i\}|\cup_{j <i} \{v_j\})} \cdot  c_e(\{v_i\}|\cup_{j <i} \{v_j\})\\
&\ge (1-\alpha) \sum_{i = 1}^{|S|} c_e(\{v_i\}|\cup_{j <i} \{v_j\}) \\
&= (1-\alpha) \cdot c_e(S).
\end{align*}
The (first) inequality is due to the slope of $c_e$. Now, since $c_M(S) \le B'$, it follows that $c_e(S) \le \frac{c_M(S)}{1-\alpha} \le \frac{B'}{1-\alpha} \le B$, 
which proves the statement. 
\end{proof}

\subsection{Proof of Theorem 2}
{\bf Statement: } {\em
Let the maximal relative difference between modular costs of two peer nodes be bounded by $r$, i.e., $r \ge \max_{v \in V, v_p \in \neighbor_v}\frac{c_m(v)}{c_m(v_p)}$, and let $\gamma = \max_{v \in V}\frac{c_{m}(v)}{B'} \in (0, \frac{1}{2})$. 
Then, the output $\selectSetAlg$ of Algorithm \ref{alg_pcpgreedyalg} has the following guarantees on the utility: 
\begin{align}
f(\selectSetAlg) \ge \left (1- e^{-\frac{(1-\alpha) \cdot (1-2 \cdot \gamma)}{1+r}} \right ) \cdot f(\opt).
\end{align} 
}
\begin{proof}
Let $\opt'$ denote an optimal solution to optimization problem $\eqref{eq.opt_2}$ when budget $B' = (1-\alpha) \cdot B$, $\bar \opt$ denote an optimal solution to optimization problem $\eqref{eq.opt_2}$ when budget $B' = B$. Cleary $f(\bar \opt) \ge f(\opt')$. Furthermore, $f(\bar \opt) \ge f(\opt)$, where $\opt$ is an optimal solution to optimization problem $\eqref{eq.opt_e}$ (and $\eqref{eq.opt}$), because cost function $c_e(S)$ is lower bounded by its modular approximation $c_M(S)$. Therefore, it suffices to lower bound $f(\selectSetAlg)$ with $f(\bar \opt)$ modulated by the approximation factor in the statement of the theorem. 

Let $S_t$ represent the current solution $\selectSetAlg$ of the greedy algorithm at time step $t$ and assume w.l.o.g. that $z^*_t$ is not $NULL$.\footnote{This is true for at least $t = 1$ due to the assumption on the boundedness of modular payments, i.e. $\gamma < \frac{1}{2}$.} Due to monotonicity and submodularity of function $f$ and the fact that Algorithm \ref{alg_pcpgreedyalg} is greedy in terms of $f/c$ ratio, we have:  
\begin{align*}
f(\bar\opt) &\le_{\text{mon.}} f(\bar\opt \cup S_t) \le_{\text{sub.}} f(S_t) + \sum_{v \in \bar\opt \backslash S_t} f(v|S_t) \\
&\le_{\text{mon.:adding a peer}} f(S_t) + \sum_{v \in \bar\opt  \backslash S_t} f(\{v, v_{v, p}\}|S_t) \\
&= f(S_t) + \sum_{v \in \bar\opt  \backslash S_t} \frac{c_M(\{v, v_{v, p}\})}{c_M(\{v, v_{v, p}\})} \cdot f(\{v, v_{v, p}\}|S_t),
\end{align*}
where $v_{v, p}$ is any peer of $v$. Using the fact that Algorithm \ref{alg_pcpgreedyalg} is greedy in terms of $f/c$ ratio, we further obtain:
\begin{align*}
f(\bar\opt) &\le f(S_t) + \frac{f(z^*_t|S_t)}{c_M(z^*_t\backslash S_t)} \cdot \sum_{v \in \bar\opt  \backslash S_t} c_M(\{v, v_{v, p}\}).
\end{align*}
Now, $r \ge \max_{v,v_p\in V: v_p \in \neighbor_v}\frac{c_m(v)}{c_m(v_p)}$ and $\sum_{v \in \bar\opt} c_m(v) \le B$ give us:
\begin{align*}
f(\bar\opt) &\le f(S_t) + \frac{f(z^*_t|S_t)}{c_M(z^*_t\backslash S_t)} \cdot (1+r) \sum_{v \in \bar\opt  \backslash S_t} c_m(v)\\
&\le f(S_t) + \frac{f(z^*_t|S_t)}{c_M(z^*_t\backslash S_t)} \cdot (1+r) \cdot B
\end{align*}
By rearranging, we get:
\begin{align*}
f(z^*_t|S_t) \ge \frac{c_M(z^*_t\backslash S_t)}{(1+r) \cdot B} \cdot [f(\bar\opt) - f(S_t) ],
\end{align*}
and further since $f(S_{t+1}) - f(S_t) = f(z^*_t|S_t)$ and $f(\emptyset) = 0$:
\begin{align*}
&f(S_{t+1}) - f(\bar\opt) \ge \left (1 - \frac{c_M(z^*_t\backslash S_t)}{(1+r) \cdot B} \right ) \cdot [f(S_t) - f(\bar\opt) ]\\
&\ge ... \ge -\prod_{i = 1}^{t+1} \left (1 - \frac{c_M(z^*_{i-1}\backslash S_{i-1})}{(1+r) \cdot B} \right) \cdot f(\bar\opt),
\end{align*}
where we used an inductive argument. In other words:
\begin{align*}
f(S_{t}) \ge \left (1-\prod_{i = 1}^{t} \left (1 - \frac{c_M(z^*_{i-1}\backslash S_{i-1})}{(1+r) \cdot B} \right) \right ) f(\bar\opt), 
\end{align*}
where we for notational convenience considered $S_t$ instead of $S_{t+1}$. Now, because:
\begin{align*}
&\prod_{i = 1}^t (1+x_i) = (e^{\frac{1}{t}\sum_{i = 1}^t \ln (1+x_i)})^t \\
&\le_{\text{$e$ is convex}} \left ( \frac{1}{t}\sum_{i = 1}^t e^{\ln (1+x_i)} \right)^t 
= \left ( 1 + \sum_{i = 1}^t \frac{x_i}{t}\right)^t,
\end{align*}
for $x_i \in (-1, 1)$, we obtain:
\begin{align*}
f(S_{t}) &\ge \left (1- \left (1 - \sum_{i = 1}^{t} \frac{c_M(z^*_{i-1}\backslash S_{i-1})}{t \cdot (1+r) \cdot B} \right)^t \right ) f(\bar\opt)\\
&\ge \left (1- \left (1 - \frac{c_M(S_t)}{t \cdot (1+r) \cdot B} \right)^t \right ) f(\bar\opt),
\end{align*}
where we used the fact that $\sum_{i = 1}^{t} c_M(z^*_{i-1}\backslash S_{i-1}) = c_M(S_t)$. Finally, we transform $B$ to $B' = (1-\alpha) \cdot B$ to obtain:
\begin{align*}
f(S_{t}) \ge \left (1- \left (1 - \frac{c_M(S_t) \cdot (1-\alpha)}{t \cdot (1+r) \cdot B'} \right)^t \right ) f(\bar\opt).
\end{align*}

Now, at a certain time step $t^*$, Algorithm \ref{alg_pcpgreedyalg} cannot find any node or a pair of nodes that could be added. Because modular costs are bounded by $c_m(v) \le \gamma \cdot B'$, we know that at time step $t^*$, Algorithm \ref{alg_pcpgreedyalg} spent at least $B' - 2\cdot \gamma \cdot B'$ budget.\footnote{Otherwise, there would be a node or a pair of nodes that could be added, or all of the elements are selected, which is the optimal choice and thus trivially proves the statement.} Hence, $c_M(S_t) \ge (1-2\cdot \gamma) \cdot B'$, which gives us that:
\begin{align*}
f(S_{t^*}) &\ge \left (1- \left (1 - \frac{(1-2 \cdot \gamma) \cdot (1-\alpha)}{ t^*  \cdot (1+r)} \right)^{t^*} \right ) f(\bar\opt)\\
&\ge \left (1- e^{-\frac{(1-\alpha) \cdot (1-2 \cdot \gamma)}{1+r}} \right ) f(\bar\opt).
\end{align*} 
Finally, using the fact that $f(\bar\opt) \ge f(\opt)$ and $\selectSetAlg = S_{t^*}$ we obtain: 
\begin{align*}
f(\selectSetAlg) \ge \left (1- e^{-\frac{(1-\alpha) \cdot (1-2 \cdot \gamma)}{1+r}} \right ) \cdot f(\opt).
\end{align*} 
\end{proof}

\end{document}